\documentclass[12pt,a4paper]{article}

\usepackage{amsmath}
\usepackage{amssymb}
\usepackage{graphicx}
\usepackage{amsthm}
\usepackage{amsfonts}
\usepackage{amscd}

\newtheorem{thm}{Theorem}[section]
\newtheorem{cor}[thm]{Corollary}

\newtheorem{example}[thm]{Example}
\newtheorem{theorem}[thm]{Theorem}

\newtheorem{lemma}[thm]{Lemma}

\newtheorem{definition}[thm]{Definition}

\newtheorem{remark}[thm]{Remark}
\newtheorem{fact}[thm]{Fact}

\begin{document}
\title
{Applications of Gaussian Binomials to Coding Theory for Deletion Error Correction}

\author{
Manabu HAGIWARA
\thanks{
Department of Mathematics and Informatics,
Graduate School of Science,
Chiba University
1-33 Yayoi-cho, Inage-ku, Chiba City,
Chiba Pref., JAPAN, 263-0022
}
 \and
Justin KONG
\thanks{
Department of Mathematics,
University of Hawaii at Manoa,
2565 McCarthy Mall (Keller Hall 401A),
Honolulu, Hawaii, U.S.A., 96822
}
}

\date{}

\maketitle

\begin{abstract}
We present new applications on $q$-binomials, also 
known as Gaussian binomial coefficients.
Our main theorems determine cardinalities 
of certain error-correcting codes based on Varshamov-Tenengolts codes
and prove a curious phenomenon relating to deletion sphere for specific cases.
\end{abstract}

\section{Introduction}
$q$-binomials \cite{stanley2013algebraic}, 
also known as Gaussian 
binomial coefficients 
\cite{cohn2004projective}
are $q$-analogs of the binomial coefficients.
They are well-known and well-studied, with important 
combinatorial implications and have properties 
analogous to binomial coefficients 
\cite{dhand2014combinatorial,kim2008q-bernoulli,pak2013strict,stanley2016some}.
However, to the best of the authors' knowledge, 
they have not been considered from 
the perspective of coding theory for deletion errors.

Terminology is defined precisely in subsequent sections, 
but here we give an informal description of the descent 
moment distribution. First, a descent vector 
(also studied in \cite{levenshtein1992perfect}) 
is a binary $01$-vector 
that indicates the indices of 
descent in an associated vector. 
The moment of a $01$-vector 
(also studied in 
\cite{levenshtein1966binary,levenshtein1967asymptotically,levenshtein1992perfect,tenengolts1965correction}) is a 
summation of the product of the index 
by the value of the $01$-vector. 
A descent moment is simply the amalgam of these 
two concepts, and the ``descent moment distribution" 
of a set of vectors is the polynomial whose 
coefficients indicate the number of vectors 
having a particular descent moment in the given set.

The main contributions of this paper 
relate to a class of deletion codes.
This provides implications for calculating the 
cardinality of sets that are of interest in the theory of error-correcting codes.

The descent moment distribution in the formula of the main theorem 
above is taken over certain sets of interest.
These sets are related to 
a well-known class of sets studied by R. P. Stanley 
known as VT (Varshamov-Tenengolts) codes \cite{stanley1972study}
(also known as special cases of Levenshtein codes  
\cite{tenengolts1965correction}).
In his study, Stanley obtained an exact formula for the cardinality of 
VT codes by considering a certain moment distribution in 
conjunction with the Hamming weight. 
The formula was non-trivial, and involved the sum of 
M\"obius functions and Euler functions. 
His formula was for the original VT codes, 
but other related sets, in particular permutation and 
multi-permutation codes based on VT codes, 
do not have such formulas. Moreover, only the moment 
distribution was considered, not the descent moment 
distribution nor any relation to $q$-binomials. 

Partial results were presented at the IEEE International 
Symposium on Information Theory (ISIT) 2018 \cite{DBLP:conf/isit/HagiwaraK18}.

\section{Preliminaries and Remarks}\label{Contribution}
\subsection{Descent Moment Distribution}
Let $\mathbf{x}:= x_1 x_2 \dots x_l$ be an element of $\{A, B\}^l $,
where $\{A, B\}$ is a binary ordered set with $A < B$.
Instead of $(x_1, x_2, \dots, x_l)$, the notation $x_1 x_2 \dots x_l$ is used in this paper.
A 01-vector $\mathbf{y} := y_1 y_2 \dots y_{l-1} \in \{ 0,1 \}^{l-1}$ is called 
a \textit{descent vector} of $\mathbf{x}$ if
\[
 y_i = 
\begin{cases}
1 & (x_i > x_{i+1}), \\
0 & (\text{otherwise}).
\end{cases}
\]
We denote 
the descent vector of $\mathbf{x}$ by $\delta(\mathbf{x})$. 
Sets considered in this paper are often 
defined via conditions with descent vectors.
For a 01-vector $\mathbf{y}$,
the \textit{moment} of $\mathbf{y} := y_1 y_2 \dots y_{l-1}$ is defined as
$y_1 + 2 y_2 + \dots + (l-1) y_{l-1}$.
The moment is denoted by $\rho ( \mathbf{y})$.
Note that the moment does not belong to a binary field 
but rather it is defined as an integer, whereas $\mathbf{y}$ is a 01 vector.
For a set $C$ of binary sequences,
we introduce the following polynomial $\mathrm{DM}(C)$ of $q$ as 
our primary interest:
\[
\mathrm{DM}(C) := 
\sum_{ \mathbf{x} \in C} q^{ \rho \circ \delta ( \mathbf{x} ) },
\]
where $\rho \circ \delta( \mathbf{x}) := \rho( \delta (\mathbf{x}) )$.

\begin{remark}\label{rem:major}
It is easy to see that
$$\rho  \circ \delta( \mathbf{x}) = \sum_{1 \le i < n, x_{i} > x_{i+1} } i.$$
The right hand side is well-known 
as the \textbf{major index} of $\mathbf{x}$.

In this paper, we call $\mathrm{DM}()$ the \textbf{descent moment distribution}
for connecting coding theory,
while it is the statistic of major index.
\end{remark}

The Hamming weight distribution is an object similar to $\mathrm{DM}$
\cite{macwilliams1963theorem}.
If $A = 0$ and $B = 1$
the distribution is 
defined as
$
\sum_{ \mathbf{x} \in C} Y^{\mathrm{wt}( \mathbf{x} )},
$
where $\mathrm{wt}( \mathbf{x} )$ is the Hamming weight of $\mathbf{x}$,
i.e., the number of non-zero entries of $\mathbf{x}$.
Both distributions may be applied to obtain the cardinality $\# C$ by substituting $1$ for their variable:
\[
\# C =
\sum_{ \mathbf{x} \in C} 1
=
\sum_{ \mathbf{x} \in C} 1^{ \rho \circ \delta ( \mathbf{x} ) }
=
\sum_{ \mathbf{x} \in C} 1^{\mathrm{wt}( \mathbf{x} )}.
\]
Another related distribution with both the moment and Hamming weight is: 
$$
\sum_{ \mathbf{x} \in C} q^{\rho( \mathbf{x} )} Y^{ \mathrm{wt}( \mathbf{x} ) },
$$
which is used to obtain the Hamming weight distribution for VT (Varshamov-Tenengolts) 
codes (see \ref{subsec:Levenshtein} in  \cite{stanley1972study}).

Notice that the descent moment distribution for the union of 
disjoint sets is equal to the sum of their descent moment distributions.
\begin{lemma}\label{DM_SUM}
For $C_1, C_2 \subset \{ A, B \}^l$,
if $C_1 \cap C_2 = \emptyset$,
$$
\mathrm{DM}( C_1 \cup C_2 ) = \mathrm{DM}(C_1) + \mathrm{DM}( C_2 ).
$$
\end{lemma}
\begin{proof}
\begin{align*}
\mathrm{DM}( C_1 \cup C_2) &=
\sum_{ \mathbf{x} \in C_1 \cup C_2 } q^{ \rho \circ \delta ( \mathbf{x} ) }\\
&= 
\sum_{ \mathbf{x} \in C_1 } q^{ \rho \circ \delta ( \mathbf{x} ) }
+
\sum_{ \mathbf{x} \in C_2 } q^{ \rho \circ \delta ( \mathbf{x} ) }\\
&=
\mathrm{DM}(C_1) + \mathrm{DM}(C_2) 
\end{align*}
\end{proof}
\subsection{$q$-integer, $q$-factorial and $q$-binomial}
The notion of a $q$-analogue is a general notion in pure mathematics
for generalizing or extending a mathematical object.
For a mathematical object $f$,
another mathematical object $F(q)$ is called a $q$-analogue
of $f$ if $f = F(1)$ or $f = \lim_{q \rightarrow 1} F(q)$.

For a positive integer $i$,
the \textit{$q$-integer} $[i]$ is defined as
\[ 
[i] := \frac{q^i - 1}{q - 1}
= 1 + q + \dots + q^{i-1},
\]
and the \textit{$q$-factorial} $[i]!$ is defined as
\begin{align*}
[i]! &:= [i] [i-1] \dots [2] [1],\\
[0]! &:= 1.
\end{align*}
Using $q$-factorials, for non-negative integers $i$ and $j$,
we define the \textit{$q$-binomial} as
\[
\begin{bmatrix} i \\ j \end{bmatrix}
:=
\begin{cases}
\dfrac{[i]!}{[j]![i-j]!} & (i \ge j \ge 0),\\
0 & (\textrm{otherwise}).\\
\end{cases}
\]
$q$-binomial is also called a \textbf{Gaussian binomial coefficient}.
It is easy to see that the $q$-integer $[i]$, the $q$-factorial $[i]!$, and the $q$-binomial$ \begin{bmatrix} i \\ j \end{bmatrix}$
are $q$-analogues of the integer $i$, the factorial $i!$, and the binomial $\binom{i}{j}$ respectively.

$q$-binomials are known to correlate to certain weight distributions of lattice paths.
Let us consider the set $L$ of lattice paths from $(0,0)$ to $(j, i-j)$.
As is well-known, its cardinality $\# L$ is given by $\binom{i}{j}$.
By defining the weight $S(p)$ of a path $p$ as the number of squares
which are on the north-western side of the path,
the following is also well-known \cite{stanley2013algebraic} (see Example \ref{example:path22} below):
\begin{align}
\sum_{p \in L} q^{S(p)} =
\begin{bmatrix} i \\ j \end{bmatrix}.\label{known fact}
\end{align}
\begin{example}[Paths from $(0,0)$ to $(2,2)$]\label{example:path22}
There are 6 paths from $(0,0)$ to $(2, 2)$ (see Figure \ref{figure:smallPath}).
\begin{figure}[htbp]
\begin{center}
\includegraphics[width=11cm,bb=0 0 348 203]{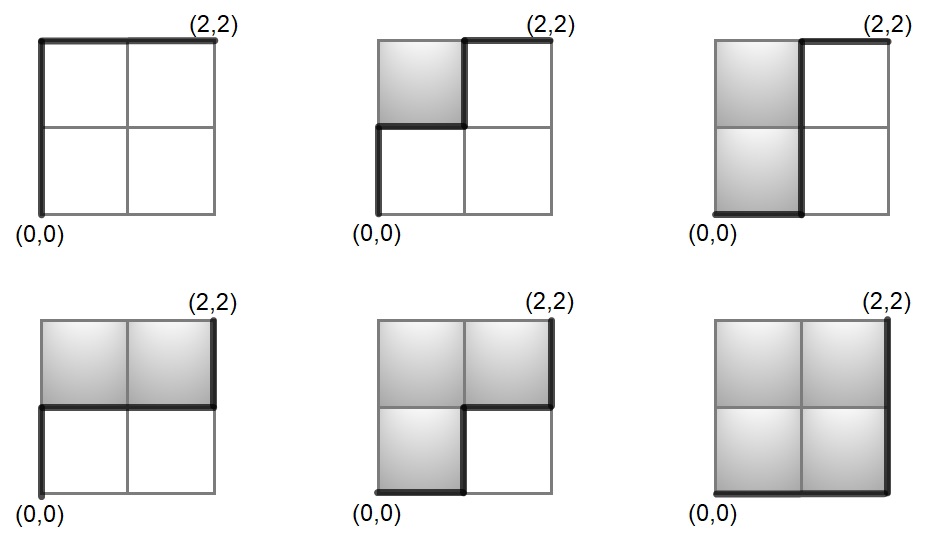} 
\caption{lattice paths from $(0,0)$ to $(2,2)$ }
\label{figure:smallPath}
\end{center}
\end{figure}
Their weights are $0$, $1$, $2$, $2$, $3$, and $4$.
Hence the weight distribution is
\[
1 + q + 2 q^2 + q^3 + q^4.
\]
On the other hand, the $q$-binomial with $i=4$ and $j=2$ is
\begin{align*}
\begin{bmatrix} 4 \\ 2 \end{bmatrix}
&=\frac{[4][3][2][1]}{[2][1][2][1]}\\
&=\frac{[4] [3]}{[2] [1]}\\
&=\frac{(1+q+q^2+q^3)(1+q+q^2)}{(1+q)(1)}\\
&=1+q+2 q^2 + q^3 + q^4.
\end{align*}
\end{example}

The following is used in the proof of Corollary \ref{theorem:Card}.

\begin{lemma}\label{app:2}
Let $\zeta$ be the $d$th primitive root of $1$
and $\zeta^{\alpha + \beta} = 1$.
\[
\lim_{q \rightarrow \zeta} \begin{bmatrix} \alpha + \beta \\ \beta \end{bmatrix}
=
\begin{cases}
\displaystyle \binom{ \frac{\alpha + \beta}{d} }{\frac{\beta}{d} } & d | \langle \alpha, \beta \rangle,\\
0                               &\text{otherwise,}
\end{cases}
\]
where $\langle \alpha, \beta \rangle$ is the greatest common divisor of $\alpha$ and $\beta$.
\end{lemma}
\begin{proof}
The assumption $\zeta^{\alpha + \beta} = 1$ implies that $d | \alpha + \beta$,
and in particular $1 \le d \le \alpha + \beta$.
The number of zero factors of $[\alpha + \beta]!$ for substituting $\zeta$ to $q$
is $\lfloor (\alpha + \beta) / d \rfloor$ and the number of zero factors of $[\alpha]! [\beta]!$ is
$\lfloor \alpha /d \rfloor + \lfloor \beta / d \rfloor$.

If $d$ does not divide $\langle \alpha, \beta \rangle$,
it implies $\lfloor (\alpha + \beta) / d \rfloor > \lfloor \alpha /d \rfloor + \lfloor \beta / d \rfloor$.
Hence $\lim_{q \rightarrow \zeta} \begin{bmatrix}  \alpha + \beta \\ \beta \end{bmatrix} = 0$.

If $d$ divides $\langle \alpha, \beta \rangle$,
it implies $\lfloor (\alpha + \beta) / d \rfloor = \lfloor \alpha /d \rfloor + \lfloor \beta / d \rfloor$.
Note that 
\begin{align*}
\lim_{q \rightarrow \zeta} \begin{bmatrix} cd \\ bd\end{bmatrix} &= \frac{c}{b},
&
\text{and }
\lim_{q \rightarrow \zeta} \begin{bmatrix} a + cd \\ a + bd\end{bmatrix}
&= 1,
\end{align*}
for $0 < a < d$.
Hence
\begin{align*}
\lim_{q \rightarrow \zeta} \begin{bmatrix} \alpha + \beta \\ \beta \end{bmatrix}
&=
\lim_{q \rightarrow \zeta}
\frac{[\alpha + \beta] [\alpha + \beta - d] \dots [d] }{([\alpha] [\alpha - d] \dots [d] )( [\beta] [\beta - d] \dots [d] ) }\\
&=
\frac{(\alpha + \beta) (\alpha + \beta - d) \dots (d) }{( (\alpha) (\alpha - d) \dots (d) ) ( (\beta) (\beta - d) \dots (d) )}\\
&=
\frac{(\frac{\alpha + \beta}{d}) (\frac{\alpha + \beta}{d} - 1) \dots (1) }{( (\frac{\alpha}{d}) (\frac{\alpha}{d}-1) \dots (1) ) ( (\frac{\beta}{d}) (\frac{\beta}{d}-1) \dots (1) ) }\\
&=
\binom{ \frac{\alpha + \beta}{d} }{ \frac{\beta}{d} }.
\end{align*}
\end{proof}
\subsection{Major Index and $q$-binomial}
For positive integers $\alpha$ and $\beta$, let $C_{\alpha, \beta}$ be the set of vectors 
with $\alpha$ entries of $A$ and $\beta$ entries of $B$.
Hence $C_{\alpha, \beta}$ consists of $\binom{\alpha + \beta}{\alpha}$ elements 
that are obtained by all permutations to $AA \dots A BB \dots B$.

The following is well-known for major index.
\begin{fact}[See \cite{stanley1997enumerative}]\label{theorem:DM}
For any positive integers $\alpha$ and $\beta$,
\begin{align*}
\mathrm{DM}( C_{\alpha, \beta} )
(
=
\sum_{0 \le m < \alpha + \beta}
\mathrm{DM}( C_{\alpha, \beta, m} )
)
=
\begin{bmatrix} \alpha + \beta \\ \beta \end{bmatrix}
\end{align*}
\end{fact}

We partition $C_{\alpha, \beta}$ into $\alpha + \beta$ subsets as follows:
\[
C_{\alpha, \beta, m}:= \{ \mathbf{x} \in C_{\alpha, \beta} \mid \rho \circ \delta ( \mathbf{x} ) \equiv m \pmod{\alpha + \beta}\}.
\]

Since $\{ C_{\alpha, \beta, m} \}_{0 \le m < \alpha + \beta}$ partitions $C_{\alpha, \beta}$, 
by Lemma \ref{DM_SUM},
we remark the following:
\begin{remark}
\[
\mathrm{DM}( C_{\alpha, \beta} ) = \sum_{0 \le m < \alpha + \beta} \mathrm{DM}( C_{\alpha, \beta, m} ).
\]
\end{remark}\label{DM_remark}

From Fact \ref{theorem:DM} and the definition of $C_{\alpha, \beta, m}$,
we have the following:
\begin{cor}
\begin{align}
\sum_{0 \le m < \alpha + \beta}
\# C_{\alpha, \beta, m} q^m
\equiv
\begin{bmatrix} \alpha + \beta \\ \beta \end{bmatrix}
\pmod{ q^{\alpha + \beta} - 1 }
.
\label{eq:main2}
\end{align}
\end{cor}

\begin{example}[$\alpha = \beta = 2$]
\label{example:alpha2M2qbinom}
As seen in Example \ref{example:path22},
\[
\begin{bmatrix} 2 + 2 \\ 2 \end{bmatrix}
=
1 + q + 2 q^2 + q^3 + q^4.\]

On the other hand,
$C_{2,2,0} = \{AABB, BABA \}$,
$C_{2,2,1} = \{BAAB \}$,
$C_{2,2,2} = \{ABAB, BBAA \}$,
and
$C_{2,2,3} = \{ABBA \}$.

Hence we verify
\[
\sum_{0 \le m < 4} \# C_{2,2, m} q^m
=
2 + q + 2 q^2 + q^3
\equiv
\begin{bmatrix} 2 + 2 \\ 2 \end{bmatrix}
\pmod{q^4 - 1}.
\]
\end{example}

\subsection{Coding Theoretic Remarks: Deletions and Partitions via VT Codes}
\label{subsec:Levenshtein}

Deletion is a combinatorial operation for a sequence.
Single deletions shorten a given sequence.
For example, a sequence $AAAA$ of length $4$ changes to the sequence $AAA$ of length $3$ 
after a single deletion.
Note that a single deletion that occurs in a string of 
consecutive repeated entries results in the same sequence 
regardless of where the deletion occurs.
Indeed, the deletions in either the 1st entry or the 2nd entry from the sequence $AABAAA$
result in the same sequence $ABAAA$.
Hence a sequence $AABAAA$ of length $6$ may be changed by a single deletion 
to one of three possible sequences of length $5$: 
$ABAAA$, $AAAAA$, or $AABAA$.

For a set $C$ of vectors,
we define the set $\mathrm{dS}(C)$ as
the set of sequences
obtained by a single deletion in $C$,
and call it the \textit{deletion sphere} of $C$.
For example, for $C := \{ AABAABB \}$,
$$
\mathrm{dS}(C) = \{ABAABB, AAAABB, AABABB, AABAAB \}.
$$

A maximal consecutive subsequence of repetitions of the same entry
is called a \textit{run}.
For a vector $\mathbf{x}$, the number of runs is denoted by $|| \mathbf{x} ||$
and is called the run number in this paper.
For example, $|| AAAA || = 1$ and $||AABAAA|| = 3$.
The run number $|| \mathbf{x} ||$ is equal to the number of sequences
that are obtained by single deletions to $\mathbf{x}$:
\begin{fact}[\cite{levenshtein1967asymptotically}]
\label{fact:sphere_run}
For any vector $\mathbf{x}$,
$$
\# \mathrm{dS}( \{ \mathbf{x} \}) = || \mathbf{x} ||.
$$
\end{fact}
Hence, the cardinality of $\mathrm{dS}( C )$ for a singleton
depends on its element.

A set $C$ is called a \textit{single deletion correcting code} if
$$
\# \mathrm{dS} (C) = \sum_{\mathbf{x} \in C} \# \mathrm{dS} ( \{ \mathbf{x} \}).
$$
This definition is equivalent to
$$
\mathrm{dS}(C) = \sqcup_{ \mathbf{x} \in C} \mathrm{dS} (\{ \mathbf{x} \}).
$$

Levenshtein showed that the following sets $\mathrm{VT}_{l-1, m}$ are single deletion correcting codes for any positive integer $l$ and 
any integer $m$
\cite{levenshtein1966binary}:
$$
\mathrm{VT}_{l-1, m} := \{ \mathbf{x} \in \{0,1 \}^{l-1} \mid \rho( \mathbf{x} ) \equiv m \pmod{l} \}.
$$
This code $\mathrm{VT}_{l-1, m}$ is called a \textit{VT code}.
The set $C_{\alpha, \beta, m}$ is written by using VT codes:
$$
C_{\alpha, \beta, m} =
\{ \mathbf{x} \in \{A,B \}^{\alpha + \beta} \mid \delta( \mathbf{x} ) \in \mathrm{VT}_{\alpha + \beta-1, m} \}.
$$
The following statement strengthens our motivation to investigate
$C_{\alpha, \beta, m}$.
The proof is a direct corollary of Lemma 3.2 in \cite{levenshtein1992perfect}. 
\begin{theorem}\label{theorem:single_deletion_error_correction}
\label{section:single_deletion_error_correction}
The set $C_{\alpha, \beta, m}$ is a single deletion correcting code.
\end{theorem}

\section{Main Contributions}
\label{Corollaries}
Our main contributions of this paper are the properties of $C_{\alpha, \beta, m}$.
Theorem \ref{theorem:Card}
and Corollary \ref{cor:maxC0} are enumerative combinatorial results
and 
Theorem \ref{theorem:binarySphere}
is a coding theoretic result.
\subsection{Cardinality of $C_{\alpha, \beta, m}$}
\begin{theorem}\label{theorem:Card}
For any $\alpha, \beta$,
\[
\# C_{\alpha, \beta, m}
=
\frac{1}{\alpha + \beta}
\sum_{ d | \langle \alpha, \beta \rangle } \binom{ \frac{\alpha + \beta}{d} }{ \frac{\beta}{d} }
 \mu( \frac{ d }{ \langle d, m \rangle } ) \frac{ \phi(d) }{ \phi( d / \langle d, m \rangle )},
\]
where $\mu$ is the m\"obius function, $\phi$ is the Euler function,
and $\langle d, m \rangle$ is the greatest common divisor of $d$ and $m$.

In particular,
\[
\# C_{\alpha, \beta, 0} =
\frac{1}{\alpha + \beta} \sum_{ d | \langle \alpha, \beta \rangle } \binom{ \frac{\alpha + \beta}{d} }{ \frac{\beta}{d} }
\phi(d).
\]
\end{theorem}
\begin{proof}
Applying Eq. (\ref{eq:main2}), the second half of 
Fact \ref{theorem:DM}, 
we analyze the $q$-binomial 
$
Q(q):=
\begin{bmatrix}
\alpha + \beta\\
\beta
\end{bmatrix}\pmod{q^{\alpha + \beta} -1}.
$ 
Since the polynomial $Q(q)$
is of degree at most $\alpha + \beta$,
$Q(q)$ is determined by $\alpha + \beta$ different points of a complex field $\mathbb{C}$,
for example
the elements of the set $Z := \{1, \zeta, \zeta^2, \cdots, \zeta^{\alpha + \beta -1} \}$ of $\alpha + \beta$th roots of 1.

By Lemma \ref{app:2},
$Q(q)$ may be written as
\[
Q(q)=\displaystyle \sum_{d| \langle \alpha, \beta \rangle }
\binom{\frac{\alpha + \beta }{d}}{\frac{\beta}{d}}
h_d(q),
\]
where $h_t(q)$ is a polynomial such that: 
1) the degree is at most $\alpha + \beta -1$,
2) $h_t(z)=1$ for a primitive $t$th root of $1$,
3) $h_t(z)=0$
for $z \in Z$ but not a primitive $t$th root. Indeed, 
\[
h_t(q)=\dfrac{1}{\alpha + \beta} \cdot \displaystyle \sum_{z_0 \in Z_t}\dfrac{q^{\alpha + \beta}-1}{q-r_0}\cdot z_0^{-1},
\]
where $Z_t$ is the set of primitive $t$th roots of $1$.

Since
\[
\displaystyle\sum_{z_0\in Z_t}\dfrac{q^{\alpha + \beta}-1}{q-r_0}\cdot z_0^{-1}
=\displaystyle\sum_{z_0\in Z_t}\displaystyle\sum_{ 0 \leq m < \alpha + \beta}z^m_{0}q^m
\]
and
\[
\displaystyle\sum_{z_0\in Z_t} z_0^m 
= \mu \left(\dfrac{t}{ \langle t,m \rangle }\right)\dfrac{\phi(t)}{\phi (  t / \langle t,m \rangle )},
\]
we have
\begin{align*}
&Q(q)=\\
&\displaystyle\sum_{0 \leq m < \alpha + \beta}
\left\{
\dfrac{1}{\alpha + \beta}\displaystyle\sum_{d| \langle \alpha, \beta \rangle }
\binom{\frac{\alpha + \beta}{d}}
{\frac{\beta}{d}}
\mu\left(\dfrac{d}{ \langle d,m \rangle }\right)\dfrac{\phi(d)}{\phi (  d / \langle d,m \rangle )}
\right\}q^m.
\end{align*}
Hence by observing the coefficient of $q^m$,
\[
\# C_{\alpha, \beta, m}=\dfrac{1}{\alpha + \beta}\displaystyle \sum_{d | \langle \alpha, \beta \rangle}
\binom{\frac{\alpha + \beta}{d}}
{\frac{\beta}{d}}
\mu\left(\dfrac{d}{ \langle d,m \rangle }\right)\dfrac{\phi(d)}{\phi\left(\frac{d}{\langle d,m \rangle }\right)}.
\]

Setting $m := 0$ implies
$
\mu\left(\dfrac{d}{ \langle d,m \rangle }\right) = 1
$
and
$
\phi\left(\frac{d}{\langle d,m \rangle }\right) = 1.
$
Thus we obtain the formula for $\# C_{\alpha, \beta, 0}$.
\end{proof}

\begin{cor}
\label{cor:maxC0}
For any $\alpha, \beta$ and $m$,
$$
\# C_{\alpha, \beta, m} \le \# C_{\alpha, \beta, 0}.
$$
\end{cor}
\begin{proof}
Note that
\[
\mu( d / \langle d,a \rangle ) = -1, 0, 1 
\text{ and }
\mu( d / \langle d,0 \rangle ) = 1.
\]
This implies $\mu (  d /  \langle d,m \rangle ) \leq 1 = \mu (  d / \langle d,0 \rangle )$.
Similarly 
$\dfrac{\phi(d)}{\phi ( d / \langle d,m \rangle )} \leq 
\phi(d) = 
 \dfrac{\phi(d)}{\phi ( d / \langle d, 0 \rangle) )}$ holds.
Therefore
\[
\# C_{ \alpha, \beta, m} \le \dfrac{1}{\alpha + \beta}\displaystyle \sum_{d | \langle \alpha, \beta \rangle}
\binom{\frac{\alpha + \beta}{d}}
{\frac{\beta}{d}}
\phi(d)
=\# C_{\alpha, \beta, 0}.
\]
\end{proof}

\subsection{Deletion Sphere in the Case $\alpha = \beta$}\label{subsec:deletionSphere}
In this subsection we briefly discuss a curious phenomenon 
relating the cardinality of $C_{\alpha,\beta, m}$ and deletion spheres 
when $\alpha = \beta$. 

As is mentioned in Fact \ref{fact:sphere_run},
the cardinality of $\mathrm{dS}( C )$ for a singleton
depends on its element.
However, we have the following: 

\begin{theorem}\label{theorem:dS_a2_beta}\label{theorem:binarySphere}\label{fact:permutationSphere}
Assume $\alpha = \beta$, and set $\gamma := \alpha (= \beta)$. Then 
\[
\# \mathrm{dS} ( C_{\gamma, \gamma, m} )
= (\gamma + 1) \# C_{\gamma, \gamma, m}.\]
\end{theorem}

\begin{example}[Case $\alpha = \beta = 2$]
As we have seen in Example \ref{example:alpha2M2qbinom},
$C_{2,2,0} = \{ AABB, BABA \}$.
Hence $\mathrm{dS}( C_{2,2,0} ) = || AABB || + || BABA || = 2 + 4 = 3 \# C_{2,2,0}$.
Similarly, $\mathrm{dS}( C_{2,2,1} ) = || BAAB || = 3 = 3 \# C_{2,2,1}$.
\end{example}

\begin{definition}[$R_r(q)$]
For an integer $r$,
let us define
\begin{align*}
R_r (q) &:= \displaystyle \sum_{\substack{\mathbf{x}\in C_{\alpha, \beta} \\ || \mathbf{x} || =r}}q^{\rho \circ \delta(\mathbf{x}) }\\
&= \mathrm{DM}( \{ \mathbf{x} \in C_{\alpha, \beta} \mid || \mathbf{x} || = r \} ),
\end{align*}
where $|| \mathbf{x} ||$ is the run number of $\mathbf{x}$.
\end{definition}

\begin{lemma}\label{lemma:DM,R_r}
\begin{align}
\mathrm{DM}( C_{\alpha, \beta}) &= \sum_{2 \le r \le \alpha + \beta} R_r (q),
\label{eq:DM1}
\end{align}
and
\begin{align}
\sum_{0 \le m < \alpha + \beta} \# \mathrm{dS}( C_{\alpha, \beta, m} ) q^m &\equiv
\sum_{2 \le r \le \alpha + \beta} r R_r (q) \pmod{ q^{\alpha + \beta} - 1}.
\label{eq:DM2}
\end{align}
\end{lemma}
\begin{proof}
Since the run number of an element of $C_{\alpha, \beta}$ is greater than or equal to $2$ and
is at most $\alpha + \beta$,
Eq. (\ref{eq:DM1}) holds.

\begin{align*}
\sum_{0 \le m < \alpha + \beta} \# \mathrm{dS}( C_{\alpha, \beta, m} ) q^m
&= \sum_{0 \le m < \alpha + \beta} \big( \sum_{ \mathbf{x} \in C_{\alpha, \beta, m} } || \mathbf{x} || q^m \big) \\
&\equiv \sum_{0 \le m < \alpha + \beta} ( \sum_{ \mathbf{x} \in C_{\alpha, \beta, m} } || \mathbf{x} || q^{\rho \circ \delta( \mathbf{x} )} )\\
&= \sum_{ \mathbf{x} \in C_{\alpha, \beta}} || \mathbf{x} || q^{\rho \circ \delta( \mathbf{x} )} \\
&=\sum_{2 \le r < \alpha + \beta}( \sum_{ \mathbf{x} \in C_{\alpha, \beta }, || x || = r } r q^{\rho \circ \delta( \mathbf{x} )} )\\
&=\sum_{2 \le r < \alpha + \beta}r ( \sum_{ \mathbf{x} \in C_{\alpha, \beta }, || x || = r } q^{\rho \circ \delta( \mathbf{x} )} )\\
&=
\sum_{2 \le r \le \alpha + \beta} r R_r (q) \pmod{ q^{\alpha + \beta} - 1}.
\end{align*}
Hence Eq. (\ref{eq:DM2}) holds.
\end{proof}

These two relations above will be used for the proof of
Theorem \ref{theorem:binarySphere}.
As preparation, we show the following:

\begin{lemma}
\label{lemma:R_r}
\[
R_r(q) = 
q^{\lfloor \frac{(r-1)^2}{4}\rfloor}
\begin{bmatrix}
\alpha-1\\
\lfloor \frac{r-2}{2}\rfloor
\end{bmatrix}
\begin{bmatrix}
\beta-1\\
\lfloor \frac{r-1}{2}\rfloor
\end{bmatrix}
+
q^{\beta}q^{\lfloor \frac{(r-1)^2}{4}\rfloor}
\begin{bmatrix}
\alpha -1  \\ \lfloor \frac{r-1}{2}\rfloor
\end{bmatrix}
\begin{bmatrix}
\beta -1 \\ \lfloor \frac{r-2}{2} \rfloor
\end{bmatrix}
,
\]
where $\lfloor \epsilon \rfloor$ denotes the maximal integer that does not exceed $\epsilon$.
\end{lemma}
\begin{proof}
For the sake of brevity, we only show the case when $r$ is even. The odd case is similarly proven.
Any element of $C_{\alpha,\beta}$ with $r$ runs, where $r$ is even, has one of the following two forms:
\[
\overbrace{AA \dots A}^{a_1}
\overbrace{BB \dots B}^{b_1}
\overbrace{AA \dots A}^{a_2}
\overbrace{BB \dots B}^{b_2}
\dots \dots
\overbrace{AA \dots A}^{a_{r/2}}
\overbrace{BB \dots B}^{b_{r/2}},
 \eqno(*) 
\] 
\[
\overbrace{BB \dots B}^{b_1}
\overbrace{AA \dots A}^{a_1}
\overbrace{BB \dots B}^{b_2}
\overbrace{AA \dots A}^{a_2}
\dots \dots
\overbrace{BB \dots B}^{b_{r/2}},
\overbrace{AA \dots A}^{a_{r/2}},
 \eqno(**) 
\]
where $a_j$ ($b_j$) denotes the length of the $j$th run with entry A (B), and 
$1 \le j \le r/2$. 
For case (*), the descent moment is
\[
(a_1+b_1)+(a_1+a_2+b_1+b_2)+\cdots +(a_1+\cdots +a_{\frac{r}{2}-1}+b_1+\cdots+b_{\frac{r}{2}-1}),
\]
and for case ($\ast\ast$), the descent moment is 
\[
(b_1)+(a_1+b_1+b_2)+\cdots +(a_1+\cdots +a_{\frac{r}{2}-1}+b_1+\cdots+b_{\frac{r}{2}}).
\]
Hence
\begin{align*}
R_r(q)
&=\displaystyle \sum_{\substack{1\leq a_1,a_2,\cdots,a_{\frac{r}{2}}\\a_1+\cdots a_{\frac{r}{2}}=\alpha}}
\displaystyle\sum_{\substack{1\leq b_1,b_2,\cdots ,b_{\frac{r}{2}}\\b_1+\cdots +b_{\frac{r}{2}}=\beta}}
(
q^{\mathcal{A}^{*}}
q^{\mathcal{B}^{*}}
+
q^{\mathcal{A}^{**}}
q^{\mathcal{B}^{**}}
) \\
&=
\displaystyle\sum
q^{\mathcal{A}^{*}}
\displaystyle\sum
q^{\mathcal{B}^{*}}
+
\displaystyle\sum
q^{\mathcal{A}^{**}}
\displaystyle\sum
q^{\mathcal{B}^{**}},
\end{align*}
where
$\mathcal{A}^{*} = a_1 + (a_1 + a_2) + \dots + (a_1 + a_2 + \dots + a_{r/2 - 1})$, 
$\mathcal{A}^{**} = \mathcal{A}^{*}$, \\
$\mathcal{B}^{*} = b_1 + (b_1 + b_2) + \dots + (b_1 + b_2 + \dots + b_{r/2 - 1})$,
and
$\mathcal{B}^{**} = b_1 + (b_1 + b_2) + \dots + (b_1 + b_2 + \dots + b_{r/2})$.

For calculating $\sum q^{\mathcal{A}^{*}}$,
let us
define a bijection,
depicted by Figure \ref{figure:shiftedPath},
from the set of sequences $a_1, a_2, \dots, a_{r/2} \ge 1$
to the set of lattice paths $p$ from $(0,0)$ to $(\alpha - r/2, r/2 - 1)$.
Note that $\mathcal{A}^* = S(p) + (1 + 2 + \dots + (r/2 -1) ) = (r/2) (r/2 - 1)/ 2 + S(p)$.

\begin{figure}[htbp]
\includegraphics[width=11cm,bb=0 0 383 267]{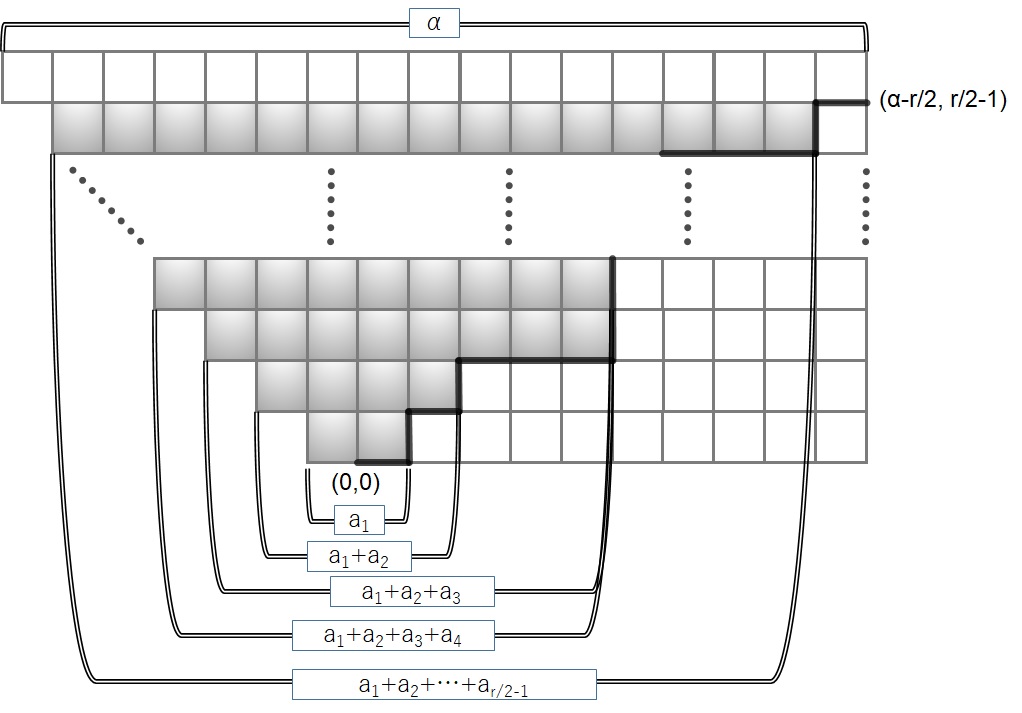} 
\caption{bijection between the set of lattice paths from $(0,0)$ to $(\alpha-r/2,r/2 -1)$
and the set of sequences $a_1, a_2, \dots, a_{r/2} \ge 1$}
\label{figure:shiftedPath}
\end{figure}

Therefore by Eq. (\ref{known fact}),
\begin{align*}
\displaystyle \sum
q^{\mathcal{A}^* }
=
q^{\frac{(r/2)(r/2 - 1)}{2}}
\begin{bmatrix}
\alpha-1\\
r/2-1
\end{bmatrix}
=
q^{\frac{(r/2)(r/2 - 1)}{2}}
\begin{bmatrix}
\alpha-1\\
\lfloor \frac{r-2}{2}\rfloor
\end{bmatrix}.
\end{align*}
Similarly
\begin{align*}
\displaystyle \sum
q^{\mathcal{B}^* }
\; = \;
q^{(r/2)(r/2 - 1)/2}
\begin{bmatrix}
\beta-1\\
r/2-1
\end{bmatrix}
\;=\;
q^{\frac{(r/2)(r/2 - 1)}{2}}
\begin{bmatrix}
\beta-1\\
\lfloor \frac{r-1}{2}\rfloor
\end{bmatrix}.
\end{align*}
Note that in the previous two equations we have used 
both $\lfloor (r-2)/2\rfloor$ and $\lfloor (r-1)/2\rfloor$ to 
represent $r/2-1$, which is permissible since $r$ is even. 
The choices were made so that the end result is consistent 
with the case when $r$ is odd. 
The previous two equations imply that

\begin{align*}
\displaystyle \sum
q^{\mathcal{A}^* }
\displaystyle \sum
q^{\mathcal{B}^* }
&\; = \;
q^{(r/2)(r/2 - 1)}
\begin{bmatrix}
\alpha-1\\
\lfloor \frac{r-2}{2}\rfloor
\end{bmatrix}
\begin{bmatrix}
\beta-1\\
\lfloor \frac{r-1}{2}\rfloor
\end{bmatrix} \\
&\; = \;
q^{ \lfloor \frac{(r-1)^2}{4} \rfloor }
\begin{bmatrix}
\alpha-1\\
\lfloor \frac{r-2}{2}\rfloor
\end{bmatrix}
\begin{bmatrix}
\beta-1\\
\lfloor \frac{r-1}{2}\rfloor
\end{bmatrix}.
\end{align*}

By a similar argument, we can show
\[
\displaystyle \sum
q^{\mathcal{A}^{**} }
\displaystyle \sum
q^{\mathcal{B}^{**} }\\
=
q^{\beta} q^{ \lfloor \frac{(r-1)^2}{4} \rfloor }
\begin{bmatrix}
\alpha-1\\
\lfloor \frac{r-1}{2}\rfloor
\end{bmatrix}
\begin{bmatrix}
\beta-1\\
\lfloor \frac{r-2}{2}\rfloor
\end{bmatrix}.
\]
Hence for even $r$,
\begin{align*}
R_r (q)
&=
q^{ \lfloor \frac{(r-1)^2}{4} \rfloor }
\begin{bmatrix}
\alpha-1\\
\lfloor \frac{r-2}{2}\rfloor
\end{bmatrix}
\begin{bmatrix}
\beta-1\\
\lfloor \frac{r-1}{2}\rfloor
\end{bmatrix}
 +
q^{\beta} q^{ \lfloor \frac{(r-1)^2}{4} \rfloor }
\begin{bmatrix}
\alpha-1\\
\lfloor \frac{r-1}{2}\rfloor
\end{bmatrix}
\begin{bmatrix}
\beta-1\\
\lfloor \frac{r-2}{2}\rfloor
\end{bmatrix}.
\end{align*}
As mentioned at the beginning of the proof, the case 
when $r$ is odd is similarly proven. 
\end{proof}

The following is the key lemma to prove 
Theorem \ref{theorem:dS_a2_beta}.
It states that
a sort of symmetry of $R_r (q)$ on $r$ holds
by
the assumption $\alpha = \beta$ and 
considering $\pmod{q^{2\gamma} - 1}$.

\begin{lemma}\label{app:1}
For $2 \le r \le 2\gamma$,  
$$
R_{r}(q)\equiv R_{2\gamma + 2 - r}(q) \pmod{q^{2 \gamma }-1}.
$$
\end{lemma}
\begin{proof}
By Lemma \ref{lemma:R_r},
\begin{align*}
R_{2 \gamma + 2 - r}(q)
&=(1+q^\gamma)
q^{\lfloor \frac{(2 \gamma + 2 - r-1)^2}{4}\rfloor}
\begin{bmatrix}
\gamma-1\\
\lfloor \frac{2 \gamma + 2 - r-1}{2}\rfloor
\end{bmatrix}
\begin{bmatrix}
\gamma-1\\
\lfloor \frac{2 \gamma + 2 - r-2}{2}\rfloor
\end{bmatrix}\\
&=
(1+q^\gamma)
q^{\gamma^2 - \gamma(r-1) + \lfloor \frac{(r-1)^2}{4}\rfloor}
\begin{bmatrix}
\gamma-1\\
\gamma-1 - \lfloor \frac{r-1}{2}\rfloor
\end{bmatrix}
\begin{bmatrix}
\gamma-1\\
\gamma-1
-
\lfloor \frac{r-2}{2}\rfloor
\end{bmatrix}\\
&=
(q^{\gamma (\gamma-r+1)} + q^{\gamma (\gamma-r+2)})
q^{\lfloor \frac{(r-1)^2}{4}\rfloor}
\begin{bmatrix}
\gamma-1\\
\lfloor \frac{r-1}{2}\rfloor
\end{bmatrix}
\begin{bmatrix}
\gamma-1\\
\lfloor \frac{r-2}{2}\rfloor
\end{bmatrix}.
\end{align*}

$q^{\gamma (\gamma-r+1)} + q^{\gamma (\gamma-r+2)} \equiv 1 + q^\gamma \pmod{q^{2\gamma} - 1}$ holds
from
$\{ \gamma-r+1, \gamma-r+2 \} \equiv \{ 0, 1 \} \pmod{2}$.
\end{proof}

\begin{proof}[Proof of Theorem \ref{theorem:dS_a2_beta}]
Define $G_m (q)$ as
$$
G_m (q) :=\displaystyle \sum_{\mathbf{x}\in C_{\gamma, \gamma, m}} || \mathbf{x}|| q^{\rho \circ \delta(\mathbf{x})},
\text{ for }0 \le m < 2 \gamma.$$
Then 
\begin{align}
\#{\rm dS}(C_{\gamma, \gamma, m})\cdot q^m
&=
 \displaystyle \sum_{\mathbf{x}\in C_{\gamma, \gamma, m}} || \mathbf{x}|| q^{m} \nonumber \\
&\equiv
G_m(q) \pmod{q^{2 \gamma} - 1}. \label{eq:lastEq}
\end{align}

Hence the proof is done by showing 
$$
G_m (q) \equiv (\gamma+1) \mathrm{DM}( C_{\gamma, \gamma, m}) \pmod{q^{2 \gamma}-1}.
$$

To this end, it is enough to show
\[
\sum_{0 \leq m < 2 \gamma} G_m (q)
\equiv (\gamma + 1){\rm DM}( C_{\gamma, \gamma}) \pmod{q^{2 \gamma}-1}.
\]

Finally we have
\begin{align*}
&\sum_{0 \leq m < 2 \gamma} G_m (q)\\
&\equiv \sum_{0 \leq m < 2 \gamma} 
   \#{\rm dS}(C_{\gamma, \gamma, m})\cdot q^m & (\text{by Eq. (\ref{eq:lastEq})})\\
&=
   \displaystyle \sum_{2 \leq r \leq 2 \gamma} r R_r(q)
   & (\text{by Eq. (\ref{eq:DM2})})\\
&=
\frac{1}{2}
\displaystyle \sum_{2 \leq r \leq 2 \gamma} r R_r(q)
         + 
\frac{1}{2}         
         \displaystyle \sum_{2 \leq r \leq 2 \gamma} r R_r(q)\\
&\equiv
\frac{1}{2} \displaystyle \sum_{2 \leq r \leq 2 \gamma} r R_r (q)
+ \frac{1}{2} \displaystyle \sum_{2 \leq r \leq 2 \gamma} r R_{2 \gamma + 2 - r} (q)
 & (\text{by Lemma \ref{app:1}})\\
&=
\frac{1}{2} \displaystyle \sum_{2 \leq r \leq 2 \gamma} r R_r(q)
+ \frac{1}{2} \displaystyle \sum_{2 \leq r \leq 2 \gamma} (2 \gamma + 2-r) R_r(q) \\
&=
\frac{1}{2} \sum_{2 \leq r \leq 2 \gamma} (2\gamma +2) R_r(q)\\
&=
(\gamma+1) \sum_{2 \leq r \leq 2 \gamma} R_r(q)\\
&=
(\gamma+1) \mathrm{DM}( C_{\gamma, \gamma} ) \pmod{ q^{2 \gamma} - 1 }.
 & (\text{by Eq. (\ref{eq:DM1})})
\end{align*}
\end{proof}

\section{Conclusion}
\label{Conclusion}

In this paper we proved a relationship between 
descent moment distributions and $q$-binomials. 
To accomplish this, we employed a lattice-path approach 
to prove pertinent lemmas.  
The relationship between descent moment distributions and 
$q$-binomials was then applied to determine the 
cardinality of $C_{\alpha, \beta, m}$. 

We have seen how the descent moment distribution 
has some interesting properties and may provide 
insights into other problems. 
Thus further investigation into descent moment distributions, 
especially as it relates to combinatorics, 
is a logical future research direction. 
Below we state two open questions 
regarding the subject. 

The natural open question 
is to extend the main results of this paper 
to ternary (or more) and then arbitrary $q$-multinomials. 
That is, the initial part of this open question is to prove a similar relationship 
for the descent moment distributions of ternary subsets of $\{A, B, C\}^l$ 
with fixed multiplicities of $A$, $B$, and $C$.

\bibliographystyle{plain}
\bibliography{reference}
\end{document}